\documentclass[aps,prb]{revtex4}

\usepackage{amsthm}
\usepackage{amsfonts}

\newtheorem{definition}[]{Definition}
\newtheorem{proposition}[]{Proposition}
\newtheorem{remark}[]{Remark}

\begin{document}

\title{Joint moments of proper delay times}

\author{Angel M. Mart\'{\i}nez-Arg\"uello}

\author{Mois\'es Mart\'{\i}nez-Mares}
\affiliation{Departamento de F\'isica, Universidad Aut\'onoma
Metropolitana-Iztapalapa, Apartado Postal 55-534, 09340 M\'exico Distrito
Federal, Mexico}

\author{Julio C. Garc\'ia}
\affiliation{Departamento de Matem\'aticas, Universidad Aut\'onoma
Metropolitana-Iztapalapa, Apartado Postal 55-534, 09340 M\'exico Distrito
Federal, Mexico}

\date{\today}

\begin{abstract}
We calculate negative moments of the $N$-dimensional Laguerre distribution for
the orthogonal, unitary, and symplectic symmetries. These moments correspond to
those of the proper delay times, which are needed to determine the statistical
fluctuations of several transport properties through classically chaotic
cavities, like quantum dots and microwave cavities with ideal coupling. 
\end{abstract}


\maketitle

\section{Introduction}
\label{sec:intro}

The generalized Laguerre ensemble appears in the context of chaotic
scattering~\cite{Frahm1,Frahm2} as being the joint distribution of the
reciprocals of the eigenvalues of the Wigner-Smith time delay matrix, known as
the proper delay times. 

The Wigner-Smith time delay matrix was introduced by Smith~\cite{Smith} as a
multi-channel generalization of the concept of delay time suffered by a wave
packet, due to the interaction with a scattering potential, introduced by Wigner
in the one channel situation.~\cite{Wigner} It is an $N\times N$ Hermitian
matrix, where $N$ is the number of open modes (or channels), whose eigenvalues,
the proper delay times, represent individual delay times on the
channels.~\cite{Frahm1,Frahm2} This time delay matrix is defined in terms of
the energy derivative of the scattering matrix $S$, which is a fundamental
entity in the description of scattering processes, and many transport properties
in open systems, by relating the outgoing plane wave amplitudes to the incoming
ones into the $N$ channels.~\cite{Mellobook} Therefore, the Wigner-Smith time
delay matrix is very important in the quantification of the transport
properties which depend on the derivative of the scattering matrix with respect
to the energy or an external parameter. The activity in this field has been
increased due to the recent theoretical
investigations\cite{Feist2014,Ivanov2014,Deshmukh2014,Chacon2014} that emerged
from a measurement of a delay time in experiments of interaction of light with
matter, with attosecond precision.\cite{Schultze2010}

In chaotic scattering, one of the most important questions is the effect of the
chaotic classical dynamics of open ballistic cavities on the transport
properties (see~\onlinecite{MelloLesHouches,BeenakkerRMP,Alhassid,Fyodorov} and
references therein). For example, the parametric derivative of the conductance
through a quantum dot, which is the analogous to the level velocity in the
characterization of a mesoscopic
system,~\cite{Simons,Fyodorov1994,Taniguchi,Mirlin} fluctuates with respect to
some external parameters, that could be an applied magnetic
field,~\cite{MarcusB} the Fermi energy,~\cite{Keller} or the dot
shape,~\cite{Chan} when they are modified by a small
amount.~\cite{Huibers,vanLangen,Castano} The DC current pumped
through a quantum dot at zero bias is quantified by its fluctuations with
respect to an applied magnetic field.~\cite{Switkes,BrouwerPumping} The same
situation occurs with the parametric derivative of the transmission coefficient,
with respect to the frequency and cavity shape, in classical wave
cavities.\cite{Schanze,MartinezMares} The statistical analysis needed for the
quantification of a transport property, or its fluctuation, is very well
realized by the random-matrix theory, that reveals the universal aspects of the
problem.\cite{MelloLesHouches,BeenakkerRMP,Alhassid} Since all of these
quantities are defined in terms of the derivative of the scattering matrix this
theory leads to quantify the transport properties, or their fluctuations, by the
first moments of the proper delay times $\tau_i$'s
($i=1,\,\ldots,\,N$).\cite{Castano,Schanze,MartinezMares,Mucciolo} In other
cases, it is the distribution of the Wigner time delay, defined as the average
of the proper delay times, what is of interest. For example, it is related to
the dimensionless capacitance of a mesoscopic
capacitor\cite{Gopar,FyodorovPRE97} or to the thermopower;\cite{Frahm2} in
disordered systems it is used to characterize the classical diffusion in the
metallic regime,\cite{Ossipov1} the eigenfunction fluctuations (see
Ref.~\onlinecite{Ossipov2} and references therein), and the metal-insulator
transition.\cite{Kottos,Antonio}

A lot of work concerning statistical studies of delay times and Wigner time
delay has been developed in the last thirty years in several contexts of chaotic
and disordered
systems.~\cite{Frahm2,FyodorovPRL96,FyodorovJMP,FyodorovPRE97,Texier,Ossipov3,
Savin,Sommers,Marciani} For closed chaotic systems, those with non perfect
coupling to the $N$ open channels, the mean and variance of partial or
phase-shift times, being the energy derivative of the eigenphases of the
scattering matrix, as well as of the Wigner time delay, are well known in
absence of time reversal symmetry ($\beta=2$).\cite{FyodorovPRL96,FyodorovJMP}
The existing distribution  for partial delay times for arbitrary $N$ are given
in Refs.~\onlinecite{FyodorovJMP}, \onlinecite{Seba}, and \onlinecite{Savin} for
$\beta=2$, and in Ref.~\onlinecite{Savin} in the presence of time reversal
symmetry ($\beta=1$) and spin-rotation symmetry ($\beta=4$). Besides, the
distribution of the Wigner time delay was first calculated in
Refs.~\onlinecite{Gopar,FyodorovJMP} for the $N=1$ case, while the one for $N=2$
and $\beta=2$ was calculated, generalized, and verified to arbitrary $\beta$ in
the ideal coupling case, and related to that of arbitrary coupling in
Ref.~\onlinecite{Savin}. The joint distribution of the reciprocals of the proper
delay times, which is the Laguerre distribution, is also known for arbitraries
$N$ and $\beta$.~\cite{Frahm1,Frahm2,Savin} However, it does not exist enough
evaluations of the moments of this distribution for any
symmetry and $N$. Expressions for the density of proper delay times and
uncorrelated moments have been evaluated in Ref.~\onlinecite{Berkolaiko} in the
limit of very large $N$ and, for few and large number of channels and any
symmetry, in Refs.~\onlinecite{Mezzadri1,Mezzadri2} more recently, both cases
for perfect coupling.

Our purpose in the present paper is to evaluate several joint moments of proper
delay times for any symmetry, and determine their dependence with the number of
channels, when this number is arbitrary. Many of these moments have not been
calculated before, some of which have also importance on transport properties
through ballistic open systems. Therefore, we regard perfect coupling to the $N$
open channels. 

In the next section we introduce the Laguerre distribution and establish the
calculation of the moments we are interested in, in a general form, and
summarize some known results for the partial times. In Sect.~\ref{sec:III} we
present some definitions and properties that will help us to manage the Laguerre
distribution, and allow us to determine the moments for arbitrary $N$ and any
symmetry. Explicit calculations are performed in Sect.~\ref{sec:explicit}. We
conclude in Sect.~\ref{sec:conclusions}.

\section{Delay times and generalized Laguerre distribution}

\subsection{Proper delay times}

A symmetrized form of the Wigner-Smith time delay matrix~\cite{Frahm1} can be
written in dimensionless units as
\begin{equation}
Q = -\frac{\mathrm{i}\hbar}{\tau_H}\, S^{-1/2} 
\frac{\partial S}{\partial E}
S^{-1/2},
\end{equation}
where $E$ is the energy and $\tau_H$ is the Heisenberg time
($\tau_H=2\pi\hbar/\Delta$, with $\Delta$ the mean level spacing). The matrix
$Q$ is an $N\times N$ Hermitian matrix, whose eigenvalues $q_i$'s
($i=1,\,\ldots,\,N$) are the proper delay times measured in units of $\tau_H$;
that is, $q_i=\tau_i/\tau_H$. The distribution of the
proper delay times is known and it is given in terms of their reciprocals. If
$x_i=\tau_H/\tau_i$, the joint distribution of the $x_i$'s
is given by the Laguerre ensemble, namely~\cite{Frahm1}
\begin{equation}
\label{eq:Laguerre}
P_{\beta}(x_1,\ldots,x_N) = C_N^{(\beta)} 
\prod_{a<b}^N \left| x_b - x_a \right|^{\beta} 
\prod_{c=1}^N x_c^{\beta N/2}\, \mathrm{e}^{-\beta x_c/2}  ,
\end{equation}
where $\beta$ characterizes the universal statistics in Dyson's
scheme:~\cite{Dyson} $\beta=1$ in the presence of time reversal invariance
(TRI) and integral spin or half-integral spin plus rotational symmetry,
$\beta=4$ for TRI, half-integral spin and no rotational symmetry, and $\beta=2$
in the absence of TRI.~\cite{Mellobook} In Eq.~(\ref{eq:Laguerre}),
$C_N^{(\beta)}$ is a normalization constant defined through the condition
\begin{equation}
\label{eq:normalization-1}
\int_0^{\infty} 
P_{\beta}\left( x_1,\ldots,x_N \right) \, 
\mathrm{d}x_1\cdots \mathrm{d}x_N = 1. 
\end{equation}

Generalized moments of the proper delay times (in dimensionless units) like
$\left\langle q_1^{k_1}\cdots q_N^{k_N}\right\rangle^{(\beta)}$ are just the
negative
generalized moments of the Laguerre ensemble. That is, 
\begin{equation}
\label{eq:jointmoments}
\left\langle q_1^{k_1}\cdots q_N^{k_N}\right\rangle^{(\beta)} = 
\left\langle x_1^{-k_1}\cdots x_N^{-k_N}\right\rangle^{(\beta)} =
\int_0^{\infty} 
\frac{P_{\beta}\left( x_1,\ldots,x_N \right)}
{x_1^{k_1}\cdots x_N^{k_N}} \, \mathrm{d}x_1\cdots \mathrm{d}x_N, 
\end{equation}
for $k_j$ ($j=1,\ldots,N$) an integer number. 

\subsection{Partial delay times}

It is instructive to compare the distribution of the proper delay times with the
distribution of ``partial delay times'', defined as the energy derivative of
phase-shifts. The distribution of an individual partial delay time scaled with
$\tau_H$, $\tau_s$, for the $\beta=2$ symmetry, is given
by~\cite{FyodorovJMP,Seba}
\begin{equation}
\label{eq:Ps}
P_s(\tau_s) = \frac{1}{N!}\, \tau_s^{-N-2} \mathrm{e}^{-1/\tau_s}.
\end{equation}
In this case it is easy to evaluate the $k$th moment of the distribution, which
is the following:
\begin{equation}
\label{eq:k-partial}
\left\langle\tau_s^k\right\rangle = \frac{(N-k)!}{N!},
\end{equation}
for $k\leq N$. In particular
\begin{equation}
\left\langle\tau_s\right\rangle = \frac{1}{N} \quad\mbox{and}\quad 
\left\langle\tau_s^2\right\rangle = \frac{1}{(N-1)N}, 
\end{equation}
which is the result expressed in Refs.~\onlinecite{FyodorovJMP,FyodorovPRL96}.

\section{General expressions for the joint moments} 
\label{sec:III}

To obtain a feasible general expression for the joint moments we notice that 
\begin{equation}
\label{eq:Vandermonde}
\prod_{a<b\leq N} (x_b - x_a) = \det V_N =
\left| 
\begin{array}{ccccc}
1 & 1 & 1 & \cdots & 1 \\ 
x_1 & x_2 & x_3 & \cdots & x_N \\ 
\vdots & \vdots & \vdots & \cdots & \vdots \\
x_1^{N-1} & x_2^{N-1} & x_3^{N-1} & \cdots & x_N^{N-1} 
\end{array} \right| ,
\end{equation}
which is known as the $N$th-order Vandermonde determinant.~\cite{Gradshteyn} It
can be proved that~\cite{Birk}
\begin{equation}
\label{eq:detVN-1}
\det\, V_N =\sum_{\sigma} \textrm{sgn} \, \sigma 
\prod_{a=1}^{N} x_a^{-1+\sigma(a)} ,
\end{equation}
where $\sigma$ is a permutation that belongs to the \emph{symmetric group of
degree N},~\cite{Herstein}
\begin{equation} 
\label{eq:SN}
S_N = \left\{ \sigma\,: \left\{ 1,\dots,N \right\} \to 
\left\{ 1,\dots,N \right\} \,|\, \sigma\, \mbox{is a permutation} \right\} ,
\end{equation}
with $\textrm{sgn}\,\sigma$ the signature of $\sigma$; the identity permutation
is $\iota(a)=a$ for $a=1,\dots,N$. 

At this point is necessary to introduce some definitions and properties. 

\begin{definition}
\label{def:DeltaN1}
For $t_a\ge 0$ ($a=1,\ldots,N$), let
\begin{equation}
\label{eq:simplex-1}
\Delta_N = \left\{ \left(t_1,\ldots,t_N\right)\, :\, 
0\le t_1\le t_2\le\cdots\le t_N \right\}.
\end{equation}
\end{definition}

Therefore, the Vandermonde determinant of Eq.~(\ref{eq:Vandermonde}) becomes
positive for $(x_1,\ldots,x_N)\in\Delta_N$, so that
\begin{equation}
\label{eq:detVN-2}
\mathrm{det}\,V_N = \left| \mathrm{det}\,V_N \right| , 
\end{equation}
and the equality in Eq.~(\ref{eq:detVN-1}) can be written as
\begin{equation}
\label{eq:detVN-3}
\left| \textrm{det}\, V_N \right|^{\beta} = 
\sum_{\sigma_1,\dots,\sigma_{\beta}} 
\prod_{i=1}^{\beta} \textrm{sgn}\, \sigma_i 
\prod_{a=1}^N x_a^{-\beta+\sum_{j=1}^{\beta}\sigma_j(a)},
\end{equation}
where $\sigma_j\in S_N$ ($j=1,\ldots,\beta$).

Each permutation $\sigma\in S_N$ is associated to an unitary transformation
$\pi_{\sigma}$ in the $N$-dimensional Euclidean space:
\begin{equation}
\label{eq:pi}
\left(x_1,\ldots,x_N\right) \longrightarrow 
\pi_{\sigma} \left(x_1,\ldots,x_N \right) = 
\left[ x_{\sigma(1)},\dots,x_{\sigma(N)} \right].
\end{equation}

\begin{definition}
\label{def:DeltaN2}
For any $\eta\in S_N$ and $t_a\ge 0$ ($a=1,\ldots,N$), let 
\begin{eqnarray}
\label{eq:simplex-2}
\Delta_N^{\eta} = \pi_{\eta}^{-1} 
\left( \Delta_N \right) := \left\{ \left( t_1,\ldots,t_N\right)\,:\, 
0\le t_{\eta(1)}\le t_{\eta(2)}\le \cdots\le t_{\eta(N)} \right\}.
\end{eqnarray}
\end{definition}

We notice that $\Delta_N^{\iota}=\Delta_N$; also, for any $\eta\in S_N$
and $\left(t_1,\ldots,t_N\right)\in\Delta^{\eta}_N$, from Eq.~(\ref{eq:pi})
\begin{equation}
\label{eq:invpi}
\pi^{-1}_{\eta}\left(t_1,\ldots,t_N\right)=
\left[t_{\eta^{-1}(1)},\dots,t_{\eta^{-1}(N)}\right]\in\Delta_N
\end{equation}
and therefore
\begin{equation}
\det\, V_N\left[ \pi^{-1}_{\eta}(t_1,\ldots,t_N) \right] = 
\prod_{a<b} \left[ t_{\eta^{-1}(b)}-t_{\eta^{-1}(a)} \right] \ge 0, 
\end{equation}
where Eq.~(\ref{eq:detVN-2}) has been taken into account. Hence, 
\begin{equation} 
\label{eq:detVinvpi}
|\det\, V_N(t_1,\ldots,t_N)| = 
\det\, V_N\left[\pi^{-1}_{\eta}(t_1,\ldots,t_N)\right] .
\end{equation}
Finally, for any nonnegative measurable function $f(t_1,\ldots,t_N)$, the
Change of Variables Theorem allows us to write 
\begin{eqnarray}
\int_0^{\infty}f(t_1,\ldots,t_N) \mathrm{d}t_1\cdots \mathrm{d}t_N & = & 
\sum_{\eta\in S_N} \int_{\Delta_N^{\eta}} 
f\left( t_1,\ldots,t_N \right) 
\mathrm{d}t_1\cdots \mathrm{d}t_N 
\nonumber \\  & = &  
\sum_{\eta\in S_N} \int_{\Delta_N} f\circ \pi_{\eta}^{-1} 
\left( t_1,\ldots,t_N \right) \mathrm{d}t_1\cdots \mathrm{d}t_N.
\label{eq:CVT}
\end{eqnarray}

\subsection{Negative moments of the Laguerre distribution of $N$ variables}
\label{ssec:nmoments}

From the definition of the Vandermonde determinant (\ref{eq:Vandermonde}) and
Eq.~(\ref{eq:detVN-3}), the Laguerre distribution of Eq.~(\ref{eq:Laguerre})
can be written as
\begin{equation}
\label{eq:Laguerre-N2}
P_{\beta}(x_1,\ldots,x_N) = C^{(\beta)}_N 
\sum_{\sigma_1,\dots,\sigma_{\beta}} \, 
\prod_{i=1}^{\beta} \mathrm{sgn}\, \sigma_i 
\prod_{a=1}^N x_a^{\gamma_a} \mathrm{e}^{-\beta x_a/2},
\end{equation}
where
\begin{equation}
\label{eq:gammai}
\gamma_a := \gamma_a^{\sigma_1,\dots,\sigma_{\beta}} = 
\frac{N\beta}{2}-\beta + 
\sum_{j=1}^{\beta}\sigma_j(a) 
\quad\mbox{for}\quad a=1,\ldots,N .
\end{equation}
It is convenient to write $P_{\beta}(x_1,\ldots,x_N)$ as
\begin{equation}
\label{eq:RN}
P_{\beta}(x_1,\ldots,x_N) = C^{(\beta)}_N 
\sum_{\sigma_1,\dots,\sigma_{\beta}} \, 
\prod_{i=1}^{\beta} \textrm{sgn}\, \sigma_i
\prod_{a=1}^N \left( \frac{2}{\beta}\right)^{\gamma_a+1}\Gamma(\gamma_a+1) \,
f_{\gamma_a+1,{\beta}/2}(x_a), 
\end{equation}
where $f_{u,v}(t)$ is the probability density function of the Gamma
distribution with parameters $u$ and $v$,~\cite{Ross,Abramowitz}
\begin{equation}
\label{eq:density-Gamma}
f_{u,v}(t) = 
\frac{v^u t^{u-1} \mathrm{e}^{-ut}}{\Gamma(u)}
\end{equation}
with $\Gamma(u)$ the Gamma function.~\cite{Abramowitz} It is important to
notice that $\Gamma(\gamma_a+1)$ is well defined.

For the joint moments we are interested in, we need the auxiliary function 
\begin{equation}
\label{eq:Gdef1}
R_{\beta}(x_1,\ldots,x_N) = 
\frac{P_{\beta}(x_1,\ldots,x_N)}{x_1^{k_1}
\cdots x_N^{k_N}}, 
\end{equation}
for $k_a$ ($a=1,\ldots,N$) an integer number. Using Eq. (\ref{eq:RN}) this
function can be written as 
\begin{equation}
\label{eq:Gdef2}
R_{\beta}(x_1,\ldots,x_N) = C^{(\beta)}_N 
\sum_{\sigma_1,\dots,\sigma_{\beta}} 
\prod_{i=1}^{\beta} \textrm{sgn}\, \sigma_i 
\prod_{a=1}^N \left(\frac{2}{\beta}\right)^{\alpha_a} 
\Gamma(\alpha_a) f_{\alpha_a,\beta/2}(x_a), 
\end{equation}
where $\alpha_a=\gamma_a-k_a+1$ for $a=1,\ldots,N$, with $\gamma_a$ as in
Eq.~(\ref{eq:gammai}). Therefore, for any $\eta\in S_N$ and
$(x_1,\ldots,x_N)\in\Delta^{\eta}_N$, the properties (\ref{eq:invpi}) and
(\ref{eq:detVinvpi}) allows us to write 
\begin{equation}
R_{\beta}\circ \pi^{-1}_{\eta}(x_1,\ldots,x_N) = 
\frac{P_{\beta}(x_1,\ldots,x_N)}
{x_{\eta^{-1}(1)}^{k_1}\cdots x_{\eta^{-1}(N)}^{k_N}} = 
\frac{P_{\beta}(x_1,\ldots,x_N)}
{x_1^{k_{\eta(1)}}\cdots x_N^{k_{\eta(N)}}} .
\end{equation}
The last equality is just the definition (\ref{eq:Gdef1}), which according to
Eq.~(\ref{eq:Gdef2}), it can be written as 
\begin{equation} 
R_{\beta}\circ \pi^{-1}_{\eta}(x_1,\ldots,x_N) =
C^{(\beta)}_N \sum_{\sigma_1,\dots,\sigma_{\beta}} 
\prod_{i=1}^{\beta} \textrm{sgn}\, \sigma_i 
\prod_{a=1}^N \left(\frac{2}{\beta}\right)^{\alpha_a^{\eta}} 
\Gamma\left(\alpha_a^{\eta}\right) 
f_{\alpha_a^{\eta},\beta/2}(x_a) ,
\end{equation}
where 
\begin{equation}
\label{eq:alpha_a}
\alpha_a^{\eta}=\gamma_a-k_{\eta(a)}+1, 
\quad\mbox{for}\quad a=1,\dots,N
\end{equation}
If we integrate this equation over $\Delta_N$ and sum over $\eta$, the property 
(\ref{eq:CVT}) allows us to arrive to the desired result, namely
\begin{equation}
\label{eq:averages}
\left\langle x_1^{-k_1}\cdots x_N^{-k_N} \right\rangle^{(\beta)}  = 
C^{(\beta)}_N \sum_{\eta,\sigma_1,\dots,\sigma_{\beta}} 
F_{\theta_{\eta,N}}
\prod_{i=1}^{\beta} \mathrm{sgn}\, \sigma_i 
\prod_{a=1}^N \left(\frac{2}{\beta}\right)^{\alpha_a^{\eta}} 
\Gamma\left(\alpha_a^{\eta}\right).
\end{equation}
where $F_{\theta_{\eta,N}}$ is defined as 
\begin{equation}
\label{eq:FNtheta}
F_{\theta_{\eta,N}} = 
\int_{\Delta_N}\prod_{a=1}^{N}f_{\alpha_a^{\eta},\beta/2}(x_a)\, \mathrm{d}x_a,
\end{equation}
with 
\begin{equation}
\label{eq:theta-etaN}
\theta_{\eta,N}:= \theta_{\eta,N}^{\sigma_1,\ldots,\sigma_{\beta}}:=
\big(\alpha_1^{\eta},\dots,\alpha_N^{\eta},
\underbrace{\beta/2,
\ldots,\beta/2}_N\big).
\end{equation}

It is clear that the normalization constant $C_N^{(\beta)}$ can be obtained from
Eq.~(\ref{eq:averages}) for $k_a=0$ ($a=1,\ldots,N$), in which case the argument
in the sum of Eq.~(\ref{eq:averages}) is independent of $\eta$, such that the
sum over $\eta$ is exactly $N!$. It is given by
\begin{equation}
\label{eq:C_N}
C^{(\beta)}_N = \left[ N! \sum_{\sigma_1,\dots,\sigma_{\beta}} \, 
F_{\theta_N} \prod_{i=1}^{\beta} \mathrm{sgn}\, \sigma_i\,
\prod_{a=1}^N \left( \frac{2}{\beta}\right)^{\gamma_a+1}\Gamma(\gamma_a+1) 
\right]^{-1} ,
\end{equation}
where
\begin{equation}
\label{eq:FN}
F_{\theta_N} = \int_{\Delta_N} \prod_{a=1}^N 
f_{\gamma_a+1,\beta/2}(x_a)\, \mathrm{d}x_a .
\end{equation}

Two Remarks are worth mentioning.

\begin{remark}
\label{rem:remark0}
In Eqs.~(\ref{eq:averages}) and (\ref{eq:C_N}), it is necessary to take
into account the dependence of the parameters $\theta_N$ and $\theta_{\eta,N}$
on the permutations $\sigma_1,\ldots,\sigma_{\beta}$ but we omitted to write it
explicitly, for simplicity. 
\end{remark}

\begin{remark}
\label{rem:remark1}
Since the minimum possible value of any permutation is 1, it is easy to show
that $\gamma_a-k_a+1>0$ for $0\le k_a<\frac{N\beta}{2}+1$.
\end{remark}

This Remark ensures that the Gamma function that appears in
Eq.~(\ref{eq:averages}) is well defined. It is important to notice that the
property expressed in Remark~\ref{rem:remark1} is a general result that
restricts the values of $k_a$ ($a=1,\ldots,N$). This condition is the
most general and coincides with that of Ref.~\onlinecite{Mezzadri1} for the
particular case when $k_a=k$, for $a=1,\ldots,N$. The restriction of $k$ in
Eq.~(\ref{eq:k-partial}) appears also as a particular case. 

\subsubsection{$\beta$ an even integer number}

When $\beta$ is an even positive integer then 
$|\det V_N|^{\beta}=(\det V_N)^{\beta}$ for any
$(x_1,\ldots,x_N)\in\mathbb{R}^N$, being $\mathbb{R}^N$ the real $N$-dimensional
Euclidean space. In this case there is no need to use Eq.~(\ref{eq:CVT}) so,
in the computations of the negative moments, we obtain simpler formulas for the
joint negative moments of the delay times. 

If we integrate Eq.~(\ref{eq:Gdef2}) we have that
\begin{eqnarray}
\label{eq:calculation}
\left\langle x_1^{-k_1}\cdots x_N^{-k_N} \right\rangle^{(\beta)} & = & 
C^{(\beta)}_N \left(\frac{2}{\beta}\right)^N 
\sum_{\sigma_1,\dots,\sigma_{\beta}}
\prod_{i=1}^{\beta} \mathrm{sgn}\, \sigma_i 
\prod_{a=1}^N 
\int_0^{\infty} \left(\frac{\beta}{2}\right) 
x^{\gamma_a-k_a} \mathrm{e}^{-\beta x/2} \mathrm{d}x 
\nonumber \\ & = & 
C^{(\beta)}_N \left(\frac{2}{\beta}\right)^N 
A^{(\beta)}_N(k_1,\ldots,k_N),
\end{eqnarray}
where 
\begin{equation}
\label{eq:A(k)}
A^{(\beta)}_N(k_1,\ldots,k_N) =  
\sum_{\sigma_1,\dots,\sigma_{\beta}}
\prod_{i=1}^{\beta} \mathrm{sgn}\, \sigma_i 
\prod_{a=1}^N \left(\frac{2}{\beta}\right)^{\gamma_a-k_a} 
\Gamma(\gamma_a-k_a+1)
\end{equation}

If $k_1=\cdots=k_N=0$, then Eq.~(\ref{eq:calculation}) implies that
$C^{(\beta)}_N=(\beta/2)^N/A^{(\beta)}_N$, with
$A^{(\beta)}_N=A^{(\beta)}_N(0,\ldots,0)$.

\section{Explicit computations}
\label{sec:explicit}

We present explicit calculations with $N=1$ and 2, for $\beta=1$ and 4, and
$N=1$, 2, 3 and 4 for $\beta=2$. We make explicit, as far as possible, the
dependence on the corresponding permutations. We start our calculations with the
$\beta=1$ case since it is usually the most difficult to treat analytically.

\subsection{Explicit calculations for $\beta=1$}

For the particular case of $\beta=1$ but arbitrary $N$, Eq.~(\ref{eq:gammai}) is
written as 
\begin{equation}
\label{eq:gammai-1}
\gamma_a:= \gamma_a^{\sigma_1} = \frac{N}{2} -1 + \sigma_1(a),
\quad\mbox{for}\quad a=1,\,\ldots,\,N, 
\end{equation}
such that the moments given by Eqs.~(\ref{eq:averages}) are simplified to 
\begin{equation}
\label{eq:negativemoments-N-1}
\left\langle x_1^{-k_1}\cdots x_N^{-k_N} \right\rangle^{(1)} = 
C^{(1)}_N \sum_{\eta,\sigma_1} F_{\theta_{\eta,N}}\, \mathrm{sgn}\, \sigma_1 \,
\prod_{a=1}^N
2^{\sigma_1(a)-k_{\eta(a)}+N/2} 
\Gamma\left[\sigma_1(a)-k_{\eta(a)}+N/2\right],
\end{equation}
where $0\leq k_a<1+N/2$.

\subsubsection{The $N=1$ case}

According to Eq.~(\ref{eq:SN}), $S_1=\{\iota(1)=1\}$, that is, only the
identity belongs to the symmetric group $S_1$, such that
$\eta(1)=\sigma_1(1)=\iota(1)=1$ and $\theta_{\eta,1}=\theta_1=(3/2-k,1/2)$.
Therefore, from Eq.~(\ref{eq:negativemoments-N-1}), 
\begin{equation}
\left\langle x^{-k} \right\rangle^{(1)} = 
C^{(1)}_1 F_{\theta_1} 2^{-k+3/2}\, \Gamma\left(-k+3/2\right) = 
C^{(1)}_1 2^{-k+3/2}\, \Gamma\left(-k+3/2\right)
\end{equation}
where we have used that
\begin{equation}
F_{\theta_1} = \frac{1}{2^{3/2}\, \Gamma(3/2)} \int_0^{\infty}
t^{1/2} \mathrm{e}^{-t} \mathrm{d}t = 1.
\end{equation}

The only values that $k$ can take are 0 and 1. For $k=0$ we obtain the
normalization constant: $C_1^{(1)}=1/2^{3/2}(1/2)!$, while for $k=1$
we have  
\begin{equation}
\label{eq:x1kN1b1}
\left\langle x_1^{-k} \right\rangle^{(1)} = \left( \frac{1}{2} \right)^k 
\frac{ \left(\frac{1}{2}-k\right)!}{ \left(\frac{1}{2}\right)! }.
\end{equation}

\subsubsection{The $N=2$ case}

In this case, $S_2=\{\iota,\sigma\}$, such that $\sigma_1=\iota,\,\sigma$ with
$\iota(a)=a$ ($a=1,2$), $\sigma(1)=2$ and $\sigma(2)=1$; therefore,
$\gamma_a=\sigma_1(a)$. 
Equation Eq.~(\ref{eq:negativemoments-N-1}) for $N=2$ gives 
\begin{equation}
\left\langle x_1^{-k_1}\, x_2^{-k_2} \right\rangle^{(1)} = 
C^{(1)}_2 2^{5-k_1-k_2} \Big[ 
\left(F_{\theta_{\iota,2}^{\iota}}-F_{\theta_{\sigma,2}^{\sigma}}\right)
\Gamma(2-k_1) \Gamma(3-k_2) + 
\left(F_{\theta_{\sigma,2}^{\iota}}-F_{\theta_{\iota,2}^{\sigma}}\right)
\Gamma(3-k_1) \Gamma(2-k_2)
\Big].
\end{equation}
Here, Eq.~(\ref{eq:alpha_a}) says that
$\alpha_a^{\eta}=\sigma_1(a)-k_{\eta(a)}+1$ and Eq.~(\ref{eq:theta-etaN}) gives 
\begin{equation}
\theta_{\eta,N}^{\sigma_1} =
[\sigma_1(i)-k_{\eta(1)}+1,\sigma_1(2)-k_{\eta(2)}+1,1/2,1/2].
\end{equation}
According to Eq.~(\ref{propiedadbasica3}) of the Appendix~\ref{appendix}, we
can determine the coefficients $F_{\theta_{\eta,N}}^{\sigma_1}$ by means of the
negative binomial distribution $NB_{\alpha_2,p_2}$ with parameters $\alpha_2$
and $p_2=1/2$.~\cite{footnote} For the particular case of $N=2$
\begin{equation}
F_{\theta_{\eta,2}^{\sigma_1}} = 
1-\sum_{\ell=0}^{\sigma_1(1)-k_{\eta(1)}}
NB_{\sigma_1(2)-k_{\eta(2)}+1,1/2}(\ell) 
\end{equation}
for any $\eta$ and $\sigma_1\in S_2$ and $k_a=0,\,1$. 

For the particular case of $k_1=k_2=0$ we have that
\begin{eqnarray}
F_{\theta_2^{\iota}} & = & 1-\sum_{\ell=0}^{\iota(1)}
NB_{\iota(2)+1,1/2}(\ell)= 
1-\sum_{\ell=0}^1 NB_{3,1/2}(\ell) = \frac{11}{16}, 
\label{eq:iota2} \\
F_{\theta_2^{\sigma}} & = & 1-\sum_{\ell=0}^{\sigma(1)}
NB_{\sigma(2)+1,1/2}(\ell)= 
1-\sum_{\ell=0}^2 NB_{2,1/2}(\ell) = \frac{5}{16},
\label{eq:sigma2}
\end{eqnarray}
where we have used the Definition~(\ref{eq:NBinomial}), and the normalization
constant is $C_2^{(1)}=1/48$.

In very similar way, for $k_1=k=1$ and $k_2=0$,
\begin{equation}
\label{eq:b1N2k0}
F_{\theta_{\iota,2}^{\iota}} = \frac{7}{8},\quad 
F_{\theta_{\sigma,2}^{\iota}} = \frac{1}{2}, \quad
F_{\theta_{\iota,2}^{\sigma}} = \frac{1}{2}, \quad\mbox{and}\quad
F_{\theta_{\sigma,2}^{\sigma}} = \frac{1}{8},  
\end{equation}
such that
\begin{equation}
\label{eq:x1kN2b1}
\left\langle x_1^{-k} \right\rangle^{(1)} = \left(\frac{1}{2}\right)^k
\frac{(1-k)!}{1!}\, K_2^{(1)}(k,0), 
\end{equation}
where 
\begin{equation}
\label{eq:K2b110}
K_2^{(1)}(k,0) = 1 \quad\mbox{for}\quad k=1 .
\end{equation}

For $k_1=k_2=k=1$ we have that
\begin{equation}
\label{eq:b1N2kk}
F_{\theta_{\iota,2}^{\iota}} = \frac{3}{4}, \quad 
F_{\theta_{\sigma,2}^{\iota}} = \frac{3}{4}, \quad
F_{\theta_{\iota,2}^{\sigma}} = \frac{1}{4}, \quad\mbox{and}\quad
F_{\theta_{\sigma,2}^{\sigma}} = \frac{1}{4},
\end{equation}
and we can write the corresponding moment as 
\begin{equation}
\label{eq:x1kx2kN2b1}
\left\langle x_1^{-k}\, x_2^{-k} \right\rangle^{(1)} = 
\left(\frac{1}{2}\right)^{2k} \frac{(1-k)!}{1!}\, 
\frac{(\frac 32-k)!}{\frac 32!}. 
\end{equation}

For $\beta=1$, the calculations become much more complicated for values of $N$
larger than 2. However, some moments can be obtained numerically, some of which
are shown in Section~\ref{sub:numerics}.

\subsection{Explicit calculations for $\beta=2$}

From Eq.~(\ref{eq:gammai}) we have that $\gamma_a=N-2+\sigma_1(a)+\sigma_2(a)$
($a=1,\ldots,N$) and Eq.~(\ref{eq:A(k)}) is written as 
\begin{equation}
\label{eq:ANbeta2N}
A^{(2)}_N(k_1,\ldots,k_N) = 
\sum_{\sigma_1,\sigma_2\in S_N} 
\textrm{sgn}\,(\sigma_1\sigma_2) 
\prod_{a=1}^N [N-2+\sigma_1(a)+\sigma_2(a)-k_a]! \,;
\end{equation}
the normalization constant is $C^{(2)}_N=1/A^{(2)}_N$, with
$A_N^{(2)}=A_N^{(2)}(0,\ldots,0)$, and the moments are given by
Eq.~(\ref{eq:calculation}).

\subsubsection{The $N=1$ case}

In this special case, $\gamma_1=\sigma_1+\sigma_2-1$, such that
Eq.~(\ref{eq:ANbeta2N}) reads 
\begin{equation}
\label{eq:ANbeta2N1}
A^{(2)}_1(k) = \sum_{\sigma_1,\sigma_2\in S_1} 
\textrm{sgn}\,(\sigma_1\sigma_2) 
(\sigma_1+\sigma_2-k-1)!
= (1-k)! \, .
\end{equation}
Since $N\beta/2+1=2$, according to Remark~\ref{rem:remark1} the maximum
value for $k$ is 1. The normalization constant is obtained for $k=0$ as
$C_1^{(2)}=1/1!$. Therefore, the only moment is
\begin{equation}
\label{eq:x1kN1b2}
\left\langle x_1^{-k} \right\rangle^{(2)} = 
\frac{(1-k)!}{1!}
\quad\mbox{for}\quad k=1 .
\end{equation}

\subsubsection{The $N=2$ case}

In this case, we observe that $\gamma_a=\sigma_1(a)+\sigma_2(a)$ ($a=1,\,2$)
and Eq.~(\ref{eq:ANbeta2N}) becomes 
\begin{equation}
A^{(2)}_2(k_1,k_2) = 
(2-k_1)!(4-k_2)!+(4-k_1)!(2-k_2)!-2(3-k_1)!(3-k_2)! .
\end{equation}
Here, $N\beta/2+1=3$ such that the maximum order of the negative moments is $2$;
that is, $k_a=0,\,1,\,2$ ($a=1,\,2$). 

For $k_1=k_2=0$ we obtain the normalization constant, 
$C_2^{(2)}=1/2!(3!\cdot 2!\cdot 1!)$. The corresponding moments for
$k_1=k$ and $k_2=0$ are 
\begin{equation}
\label{eq:x1kN2b2}
\left\langle x_1^{-k}\right\rangle^{(2)} = \frac{(2-k)!}{2!}\, K_2^{(2)}(k,0) ,
\end{equation}
where
\begin{equation}
\label{eq:K22k}
K_2^{(2)}(k,0) =
\frac{1}{3!} \left[
12-6(3-k)+(3-k)(4-k) \right] .
\end{equation}
Also, for $k_1=k_2=k$ we have 
\begin{equation}
\label{eq:x1kx2kN2b2}
\left\langle x_1^{-k}x_2^{-k}\right\rangle^{(2)} = 
\frac{(3-k)!}{3!} \, \frac{(2-k)!}{2!}.
\end{equation}
The remaining joint negative moment is
\begin{equation}
\label{eq:x12x2N=2}
\left\langle x_1^{-2}x_2^{-1}\right\rangle^{(2)} = 
\frac{1!}{3!}.
\end{equation}

\subsubsection{The $N=3$ case}

For $N=3$, $\gamma_a=\sigma_1(a)+\sigma_2(a)+1$ ($a=1,\,2\,,3$) and
Eq.~(\ref{eq:ANbeta2N}) becomes 
\begin{eqnarray}
\label{eq:A32}
A^{(2)}_3(k_1,k_2,k_3) & = & 
(3-k_1)!(5-k_2)!(7-k_3)! + (3-k_1)!(7-k_2)!(5-k_3)! 
\nonumber \\ & + &
(5-k_1)!(3-k_2)!(7-k_3)! + (5-k_1)!(7-k_2)!(3-k_3)! 
\nonumber \\ & + &
(7-k_1)!(3-k_2)!(5-k_3)! + (7-k_1)!(5-k_2)!(3-k_3)! 
\nonumber \\ & - & 
2(3-k_1)!(6-k_2)!(6-k_3)! - 2(6-k_1)!(3-k_2)!(6-k_3)! 
\nonumber \\ & - &
2(6-k_1)!(6-k_2)!(3-k_3)!- 2(4-k_1)!(4-k_2)!(7-k_3)! 
 \\ & - & 
2(4-k_1)!(7-k_2)!(4-k_3)!- 2(7-k_1)!(4-k_2)!(4-k_3)! 
\nonumber \\ & + & 
2(4-k_1)!(5-k_2)!(6-k_3)! + 2(4-k_1)!(6-k_2)!(5-k_3)! 
\nonumber \\ & + &
2(6-k_1)!(5-k_2)!(4-k_3)! + 2(5-k_1)!(4-k_2)!(6-k_3)! 
\nonumber \\ & + & 
2(5-k_1)!(6-k_2)!(4-k_3)! + 2(6-k_1)!(4-k_2)!(5-k_3)!
\nonumber \\ & - &
6(5-k_1)!(5-k_2)!(5-k_3)! 
\nonumber
\end{eqnarray}
The maximum value of $k_a$ ($a=1\,,2\,,3$) is 3. If we evaluate this
expression at $k_1=k_2=k_3=0$ we obtain 
$C_3^{(2)}=1/3!(5!\cdot 4!\cdot 3!\cdot 2!\cdot 1!)$.

The result for the moments for $k_1=k$ and $k_2=k_3=0$ is given by 
\begin{equation}
\label{eq:x1kN3b2}
\left\langle x_1^{-k}\right\rangle^{(2)}  =  
\frac{(3-k)!}{3!}\, K_3^{(2)}(k,0,0),
\end{equation}
where
\begin{eqnarray}
\label{eq:K32k}
K_3^{(2)}(k,0,0) & = & 
\frac{1}{6!} 
\Big[ 5\cdot 6! -4\cdot 6!(4-k) +150\cdot 3!(4-k)(5-k) 
\nonumber \\ & - & 
5!(4-k)(5-k)(6-k) + 3!(4-k)(5-k)(6-k)(7-k) \Big].
\end{eqnarray}
In similar way, for $k_1=k_2=k$ and $k_3=0$ we have that
\begin{equation}
\label{eq:x1kx2kN3b2}
\left\langle x_1^{-k} x_2^{-k} \right\rangle^{(2)} = 
\frac{(3-k)!}{3!}\, \frac{(4-k)!}{4!}\, K_3^{(2)}(k,k,0),
\end{equation}
where
\begin{eqnarray}
\label{eq:K32kk}
K_3^{(2)}(k,k,0) & = & 
\frac{1}{6!} \Big[ 7! -6!(k-1)(5-k) 
\nonumber \\ & - & 
4\cdot 4!(4-k)(5-k)(6-k) +3!(4-k)(5-k)^2(6-k) 
\Big].
\end{eqnarray}
The moments for $k_1=k_2=k_3=k$ are given by 
\begin{equation}
\label{eq:x1kx2kx3kN3b2}
\left\langle x_1^{-k} x_2^{-k} x_3^{-k} \right\rangle^{(2)} = 
\frac{(5-k)!}{5!}\, \frac{(4-k)!}{4!} \, \frac{(3-k)!}{3!}.
\end{equation}

The remaining terms can be evaluated directly from Eq.~(\ref{eq:A32}). They are
the following:
\begin{equation}
\label{eq:x12x2N=3}
\left\langle x_1^{-2} x_2^{-1} \right\rangle^{(2)} = \frac{3!}{5!},
\quad 
\left\langle x_1^{-3} x_2^{-1} \right\rangle^{(2)} = \frac{62}{6!},
\quad\mbox{and}\quad  
\left\langle x_1^{-3} x_2^{-2} \right\rangle^{(2)} = \frac{3}{5!} \, ;
\end{equation}
also, 
\begin{equation}
\label{eq:q211}
\left\langle x_1^{-2} x_2^{-1} x_3^{-1} \right\rangle^{(2)} = 
\frac{6}{6!}, \, 
\left\langle x_1^{-2} x_2^{-2} x_3^{-1} \right\rangle^{(2)} = 
\frac{2!}{6!}, \,
\left\langle x_1^{-3} x_2^{-1} x_3^{-1} \right\rangle^{(2)} = 
\frac{10}{6!},\, 
\left\langle x_1^{-3} x_2^{-2} x_3^{-1} \right\rangle^{(2)} = 
\frac{5}{2\cdot 6!} \,;
\end{equation}

\begin{equation} 
\label{eq:q32}
\left\langle x_1^{-3} x_2^{-2} x_3^{-2} \right\rangle^{(2)} = 
\frac{1}{2\cdot 6!},
\quad
\left\langle x_1^{-3} x_2^{-3} x_3^{-1} \right\rangle^{(2)} = 
\frac{1}{4\cdot 5!}, 
\quad\mbox{and}\quad
\left\langle x_1^{-3} x_2^{-3} x_3^{-2} \right\rangle^{(2)} = 
\frac{1}{4\cdot 6!} \,.
\end{equation}

\subsubsection{The $N=4$ case}

From Eq.~(\ref{eq:gammai}), $\gamma_a=\sigma_1(a)+\sigma_2(a)+2$
($a=1\,,2\,,3\,,4$); in this case Eq~(\ref{eq:ANbeta2N}) gives 
\begin{eqnarray}
\label{eq:A42-1}
A_4^{(2)}(k_1,k_2,k_3,k_4) & = &
(4-k_1)!(6-k_2)!(8-k_3)!(10-k_4)! - (4-k_1)!(6-k_2)!(9-k_3)!(9-k_4)! 
\nonumber \\ & - &
(4-k_1)!(7-k_2)!(7-k_3)!(10-k_4)! + (4-k_1)!(7-k_2)!(9-k_3)!(8-k_4)! 
\nonumber \\ & + & 
(4-k_1)!(8-k_2)!(7-k_3)!(9-k_4)! - (4-k_1)!(8-k_2)!(8-k_3)!(8-k_4)! 
\nonumber \\ & - & 
(5-k_1)!(5-k_2)!(8-k_3)!(10-k_4)! + (5-k_1)!(5-k_2)!(9-k_3)!(9-k_4)! 
\nonumber \\ & + & 
(5-k_1)!(7-k_2)!(6-k_3)!(10-k_4)! - (5-k_1)!(7-k_2)!(9-k_3)!(7-k_4)! 
\nonumber \\ & - & 
(5-k_1)!(8-k_2)!(6-k_3)!(9-k_4)! + (5-k_1)!(8-k_2)!(8-k_3)!(7-k_4)! 
\nonumber \\ & + & 
(6-k_1)!(5-k_2)!(7-k_3)!(10-k_4)! - (6-k_1)!(5-k_2)!(9-k_3)!(8-k_4)! 
\nonumber \\ & - & 
(6-k_1)!(6-k_2)!(6-k_3)!(10-k_4)! + (6-k_1)!(6-k_2)!(9-k_3)!(7-k_4)! 
\nonumber \\ & + & 
(6-k_1)!(8-k_2)!(6-k_3)!(8-k_4)! - (6-k_1)!(8-k_2)!(7-k_3)!(7-k_4)! 
\nonumber \\ & - & 
(7-k_1)!(5-k_2)!(7-k_3)!(9-k_4)! + (7-k_1)!(5-k_2)!(8-k_3)!(8-k_4)! 
\nonumber \\ & + & 
(7-k_1)!(6-k_2)!(6-k_3)!(9-k_4)! - (7-k_1)!(6-k_2)!(8-k_3)!(7-k_4)! 
\nonumber \\ & - & 
(7-k_1)!(7-k_2)!(6-k_3)!(8-k_4)! + (7-k_1)!(7-k_2)!(7-k_3)!(7-k_4)!
\nonumber \\ & + & 
(\mbox{permutations of }\, k_1,\, k_2,\, k_3,\, k_4).
\end{eqnarray}
This equation is well defined because the maximum allowed value for
$k_a$ ($a=1,2,3,4$) is 4. 

Although the calculation of $A_4^{(2)}(k_1,k_2,k_3,k_4)$ for arbitrary set
of values of the $k_a$'s is not difficult, it consists of many terms that are
not easy to follow. Two quantities are clearly feasible: one for $k_a=0$ and the
other for $k_a=k$. This is due to the fact that the permutations of $k_a$'s in
Eq.~(\ref{eq:A42-1}) give the same terms that have been explicitly written. The
first quantity gives the normalization constant, 
$C_4^{(2)}=1/4!\left(7!\cdot 6!\cdot 5!\cdot 4!\cdot 3!\cdot 2!\cdot 1!\right)$.
The second quantity is the moment
\begin{equation}
\label{eq:x1kx2kx3kx4kN4b2}
\left\langle x_1^{-k} x_2^{-k} x_3^{-k} x_4^{-k} \right\rangle^{(2)}=
\frac{(7-k)!}{7!}\, \frac{(6-k)!}{6!}\, \frac{(5-k)!}{5!}\, \frac{(4-k)!}{4!}.
\end{equation}

Any other moment is difficult to compute arithmetically, with great effort we
arrive to the following results:
\begin{equation}
\label{eq:x1N=4}
\left\langle x_1^{-1} \right\rangle^{(2)} = \frac{3!}{4!},
\quad
\left\langle x_1^{-2}\right\rangle^{(2)} = \frac{8\cdot 2!}{5!}, 
\quad
\left\langle x_1^{-3}\right\rangle^{(2)} = \frac{16\cdot 3!}{6!}, 
\quad\mbox{and}\quad
\left\langle x_1^{-4}\right\rangle^{(2)} = \frac{59\cdot 4!}{7!}; 
\end{equation}

\begin{equation}
\label{eq:x1x2N=4}
\left\langle x_1^{-1} x_2^{-1} \right\rangle^{(2)} = 
\frac{6}{5!}, \quad 
\left\langle x_1^{-2} x_2^{-2} \right\rangle^{(2)} = 
\frac{3}{6!}\left(\frac{8\cdot 9}{6\cdot 7}\right), 
\quad\mbox{and}\quad 
\left\langle x_1^{-3} x_2^{-3} \right\rangle^{(2)} =
\frac{8\cdot 5!}{9!} ;
\end{equation}
\begin{equation}
\left\langle x_1^{-1} x_2^{-1} x_3^{-1} \right\rangle^{(2)} =
\frac{1}{5!},
\quad
\left\langle x_1^{-2} x_2^{-2} x_3^{-2} \right\rangle^{(2)} =
\frac{1}{7!},
\quad\mbox{and}\quad 
\left\langle x_1^{-3} x_2^{-3} x_3^{-3} \right\rangle^{(2)} =
\frac{51}{10!} .
\end{equation}

\subsection{Explicit calculations for $\beta=4$}

For $\beta=4$, Eq.~(\ref{eq:gammai}) says that
$\gamma_a=2N-4+\sigma_1(a)+\sigma_2(a)+\sigma_3(a)+\sigma_4(a)$
($a=1,\ldots,N$), such that Eq.~(\ref{eq:A(k)}) gives 
\begin{equation}
\label{eq:ANbeta4N}
A^{(4)}_N(k_1,\ldots,k_N) = 
\sum_{\sigma_1,\sigma_2,\sigma_3\sigma_4\in S_N} 
\textrm{sgn}\,(\sigma_1\sigma_2\sigma_3\sigma_4) 
\prod_{i=1}^N
\frac{(\gamma_i-k_i)!}{2^{\gamma_i-k_i}}.
\end{equation}
The normalization constant is given by $C_N^{(4)}=2^N/A_N^{(4)}$ and the
moments are given by Eq.~(\ref{eq:calculation}). 

\subsubsection{The $N=1$ case}

In this case, $\gamma=(\sigma_1+\sigma_2+\sigma_3+\sigma_4-k-2)!$. 
Equation~(\ref{eq:ANbeta4N}) can be written as
\begin{equation}
\label{eq:ANbeta4N1}
A^{(4)}_1(k) = 2^k\frac{(2-k)!}{2^2} ,
\end{equation}
such that $A^{(4)}_1=1/2$ and $C^{(4)}_1=2^2$. Therefore, 
\begin{equation}
\label{eq:x1kN1b4}
\left\langle x_1^{-k} \right\rangle^{(4)} = 2^k\frac{(2-k)!}{2!} 
\quad\mbox{for}\quad k=1\,,2.
\end{equation}

\subsubsection{The $N=2$ case}

Here, $\sigma_i\in S_2$ ($i=1,2,3,4$) and 
\begin{eqnarray}
\label{eq:ANbeta4N2}
A^{(4)}_2(k_1,k_2) & = & 
2^{k_1+k_2-12}\left[(4-k_1)!(8-k_2)! - 4(5-k_1)!(7-k_2)! 
+ 3(6-k_1)!(6-k_2)!\right]
\nonumber \\ & + &
(\mbox{permutations of}\, k_1\,\mbox{and}\, k_2).
\end{eqnarray}
Taking $k_1=k_2=0$, we have that 
$C_2^{(4)}=2^{14}/3!(5!\cdot 4!\cdot 3!\cdot 2!\cdot 1!)$, while for
$k_1=k_2=k$, 
\begin{equation}
\label{eq:x1kx2kN2b4}
\left\langle x_1^{-k}x_2^{-k}\right\rangle^{(4)} = 
2^{2k}\frac{(6-k)!}{6!}\, \frac{(4-k)!}{4!}, 
\quad\mbox{for}\quad k=1\,,2\,,3\,,4. 
\end{equation}
Also, for $k_1=k$ and $k_2=0$, we get 
\begin{equation}
\label{eq:x1kN2b4}
\left\langle x_1^{-k} \right\rangle^{(4)} = 
2^k\frac{(4-k)!}{4!}\, K_2^{(4)}(k,0), 
\end{equation}
where 
\begin{equation}
\label{eq:KN2b4}
K_2^{(4)}(k,0) = \frac{1}{6!\cdot 3!\cdot 2!}
\left[ 8! + 2\cdot 6!(5-k)(4-3k)-4!(12+k)(7-k)(6-k)(5-k) \right].
\end{equation}

From Eq.~(\ref{eq:ANbeta4N2}) it is easy to see that the remaining moments are
given by
\begin{equation}
\label{eq:x21}
\left\langle x_1^{-2} x_2^{-1} \right\rangle^{(4)} = \frac{1}{9}, \quad 
\left\langle x_1^{-3} x_2^{-1} \right\rangle^{(4)} = \frac{7}{45},
\quad\mbox{and}\quad
\left\langle x_1^{-3} x_2^{-2} \right\rangle^{(4)} = \frac{2}{45};
\end{equation}
\begin{equation}
\left\langle x_1^{-4} x_2^{-1} \right\rangle^{(4)} = \frac{22}{45}, \quad 
\left\langle x_1^{-4} x_2^{-2} \right\rangle^{(4)} = \frac{2}{15},
\quad\mbox{and}\quad
\left\langle x_1^{-4} x_2^{-3} \right\rangle^{(4)} = \frac{2}{45}.
\end{equation}

\subsection{Joint moments of proper delay times for arbitraries $N$ and
$\beta$} 
\label{sub:numerics}

The set of equations (\ref{eq:x1kN1b1}), (\ref{eq:x1kN2b1}), (\ref{eq:x1kN1b2}),
(\ref{eq:x1kN2b2}), (\ref{eq:x1kN3b2}), (\ref{eq:x1kN1b4}), and
(\ref{eq:x1kN2b4}); (\ref{eq:x1kx2kN2b1}), (\ref{eq:x1kx2kN2b2}),
(\ref{eq:x1kx2kN3b2}), and (\ref{eq:x1kx2kN2b4}); (\ref{eq:x1kx2kx3kN3b2}) and
(\ref{eq:x1kx2kx3kx4kN4b2}), suggest a general expression for the joint moments
for any symmetry and number of channels, which is 
\begin{equation}
\label{eq:xNlk}
\left\langle q_1^k\cdots q_m^k\right\rangle^{(\beta)} = 
\left[
\prod_{n=N}^{N+m-1} 
\left( \frac{\beta}{2} \right)^k
\frac{\left(\frac{\beta n}{2}-k\right)!}
{\left(\frac{\beta n}{2}\right)!} \right] \, 
K_N^{(\beta)}\big(\underbrace{k,\ldots,k}_m,\underbrace{0,\ldots,0}_{N-m}\big)
\quad\mbox{for}\quad m\leq N,
\end{equation}
where
$K_N^{(\beta)}\big(\underbrace{k,\ldots,k}_m,\underbrace{0,\ldots,0}_{N-m}\big)$
has a particular expression for each values of $N$ and $m$, as can be seen in
Eqs.~(\ref{eq:K2b110}), (\ref{eq:K22k}), (\ref{eq:K32k}), (\ref{eq:K32kk}), and
(\ref{eq:KN2b4}). A closed expression for this quantity is difficult to
obtain analytically, but it reduces to 1 for $m=N$, as is suggested also
by Eqs.~(\ref{eq:x1kN1b1}), (\ref{eq:x1kx2kN2b1}), (\ref{eq:x1kN1b2}),
(\ref{eq:x1kx2kN2b2}), (\ref{eq:x1kx2kx3kN3b2}), (\ref{eq:x1kx2kx3kx4kN4b2}),
(\ref{eq:x1kN1b4}), and (\ref{eq:x1kx2kN2b4}), 
\begin{equation}
K_N^{(\beta)} (k,\ldots,k)=1, 
\end{equation}
as well as for $k=1$ and $m\leq N$, 
\begin{equation}
\label{eq:K10}
K_N^{(\beta)} \big( \underbrace{1,\ldots,1}_m, 
\underbrace{0,\ldots,0}_{N-m} \big) = 1, 
\end{equation}
as can be seen by the direct evaluation of the equations just mentioned above. 

A particular case of interest is the $k$th moment for $m=1$ (single variable);
we find that 
\begin{equation}
\label{eq:qNk}
\left\langle q_1^k \right\rangle^{(\beta)} = 
\left(\frac{\beta}{2}\right)^k 
\frac{\left(\frac{\beta N}{2}-k\right)!}{\left(\frac{\beta N}{2}\right)!}\,
K_N^{(\beta)}(k,0,\ldots,0), 
\quad\mbox{for}\quad k< \frac{\beta N}{2}+1.
\end{equation}
What is interesting of this result is that the $k$th moment of a proper delay
time differs from that of a partial delay time, as can be seen if we compare
this result for $\beta=2$ with
Eq.~(\ref{eq:k-partial}).~\cite{FyodorovJMP,FyodorovPRL96} Moreover, it allows
us to generalize the $k$th moment for a partial time for any $\beta$ and $N$,
namely 
\begin{equation}
\label{eq:qNkpartial}
\left\langle \tau_s^k \right\rangle^{(\beta)} = 
\left(\frac{\beta}{2}\right)^k 
\frac{\left(\frac{\beta N}{2}-k\right)!}{\left(\frac{\beta N}{2}\right)!}
\quad\mbox{for}\quad k< \frac{\beta N}{2}+1.
\end{equation}
That is, the distribution for the partial times can be also generalized to any
$\beta$ and $N$ by replacing $N$ by $\beta N/2$ in the corresponding
distribution of the $\beta=2$ case of
Eq.~(\ref{eq:Ps}).~\cite{FyodorovJMP,FyodorovPRL96,Seba} The quantitative
difference given by the factor $K_N^{(\beta)}(k,0,\ldots,0)$ in
Eq.~(\ref{eq:qNk}), comes from the level repulsion for the proper delay times,
as happens for their corresponding density.~\cite{Frahm2} The only exception to
this rule is, of course, the case $N=1$. 

We can also notice that the factor $K_N^{(\beta)}(k,0,\ldots,,0)$ has the
following expression:
\begin{equation}
\label{eq:Kk123}
K_N^{(\beta)}(k,0,\ldots,0) = \frac{k!N^{k-1}N!}{(N+k-1)!} 
\quad\mbox{for}\quad k=1,2,3,
\end{equation}
independent of $\beta$.

From Eqs.~(\ref{eq:qNk}) and (\ref{eq:Kk123}) we obtain the following
interesting averages:
\begin{equation}
\label{eq:derivadapar}
\left\langle q_1 \right\rangle^{(\beta)} = \frac{1}{N}, 
\quad
\left\langle q_1^2 \right\rangle^{(\beta)} = \frac{\beta}{2}\, 
\frac{2}{\left(\frac{\beta N}{2}-1\right)(N+1)},
\end{equation}
and
\begin{equation}
\left\langle q_1^3 \right\rangle^{(\beta)} = 
\left(\frac{\beta}{2}\right)^2 
\frac{3!\, N}
{\left(\frac{\beta N}{2}-1\right)\left(\frac{\beta N}{2}-2\right)(N+1)(N+2)}, 
\end{equation}
which agrees with those of Ref.~\onlinecite{Berkolaiko} for
$\beta=2$ in the semiclassical limit, and of Ref.~\onlinecite{Mezzadri1} for
any $\beta$ and $N$. Both results in Eq.~(\ref{eq:derivadapar}) were reported in
Refs.~\onlinecite{Castano,MartinezMares} for $\beta=1$ and 2. 
Another case that can be easily obtained from Eqs.~(\ref{eq:xNlk}) and
(\ref{eq:K10}), which is of particular interest, is 
\begin{equation}
\label{eq:pumping}
\left\langle q_1\,q_2\right\rangle^{(\beta)} = 
\frac{1}{N(N+1)} .
\end{equation}
which was reported in Ref.~\onlinecite{Mucciolo} for $\beta=2$. 

Since the Wigner time delay is $\tau_W=\tau_H\sum_i^Nq_i/N$, its mean and
variance can be calculated from Eqs.~(\ref{eq:derivadapar}) and
(\ref{eq:pumping}) for any $\beta$ and $N$; the results are  
\begin{equation}
\tau_W = \frac{\tau_H}{N} , 
\quad
\frac{\langle\tau_W^2\rangle-\langle\tau_W\rangle^2}
{\langle\tau_W\rangle^2} = 
\frac{2}{\left(\frac{\beta N}{2}-1\right)(N+1)}, 
\end{equation}
which for $\beta=2$ reduces to those of
Refs.~\onlinecite{FyodorovPRL96,FyodorovJMP}. 

From Eqs.~(\ref{eq:x1kx2kN2b2}), (\ref{eq:x1kx2kN3b2}), (\ref{eq:K32kk}), the
second equation in (\ref{eq:x1x2N=4}), and Eq.~(\ref{eq:x1kx2kN2b4}), is
feasible to find that
\begin{equation}
K_N^{(\beta)}(2,2,0,\ldots,0) = \frac{2N(2N+1)}{(N+2)(N+3)} .
\end{equation}
Although we do not have explicit results for $\beta=1$, we have verified this
expression numerically. This result allows to obtain
\begin{equation}
\label{eq:q2q2}
\left\langle q_1^2\,q_2^2\right\rangle^{(\beta)} = 
\left(\frac{\beta}{2}\right)^4 
\frac{\left(\frac{\beta N}{2}-2\right)!}{\left(\frac{\beta N}{2}\right)!} \,
\frac{\left[\frac{\beta(N+1)}{2}-2\right]!}{\left[\frac{\beta(N+1)}{2}\right]!}
\, K_N^{(\beta)}(2,2,0,\ldots,0).
\end{equation}

\begin{table}
\begin{tabular}{cc|c|cc}
 & $\beta=1$ & $\beta=2$ & $\beta=4$ \\ \hline
\begin{tabular}{c|}
$N$ \\ \hline 
$\left\langle q_1^2 q_2 \right\rangle$ \\ 
$\left\langle q_1^3 q_2 \right\rangle$ \\ 
$\left\langle q_1^3 q_2^2 \right\rangle$ \\
$\left\langle q_1^2 q_2 q_3 \right\rangle$ \\
$\left\langle q_1^2 q_2^2 q_3 \right\rangle$
\end{tabular}
& 
\begin{tabular}{cccc}
2 & 3 & 4 & 5 \\ \hline 
--- & --- & $\frac{4\cdot 6\cdot 1!}{6!}$ &  \\ 
--- & --- & --- &  \\ --- & --- & --- & \\
--- & $\frac{2\cdot 0!}{5!}$ & $\frac{2\cdot 12\cdot 1!}{7!}$ &  \\ 
--- & --- & $\frac{2\cdot 36\cdot 1!}{8!}$ & 
\end{tabular} 
& 
\begin{tabular}{cccc}
2 & 3 & 4 & 5 \\ \hline 
$\frac{2\cdot 2\cdot 0!}{4!}$ & $\frac{2\cdot 3\cdot 1!}{5!}$ & $\frac{2\cdot
4\cdot 2!}{6!}$ & $\frac{2\cdot 5\cdot 3!}{7!}$ \\ 
--- & $\frac{2\cdot 31\cdot 0!}{6!}$ & $\frac{2\cdot 53\cdot 1!}{7!}$ &
$\frac{2\cdot 81\cdot 2!}{8!}$ \\ 
--- & $\frac{6\cdot 3\cdot 7\cdot 0!}{7!}$ & 
$\frac{6\cdot 4\cdot 9\cdot 1!}{8!}$ & $\frac{6\cdot 5\cdot 11\cdot 2!}{9!}$ \\
--- & $\frac{2\cdot 3\cdot 1!}{6!}$ & $\frac{2\cdot 4\cdot 2!}{7!}$ &
$\frac{2\cdot 5\cdot 3!}{8!}$ \\ 
--- & $\frac{2\cdot 7\cdot 1!}{7!}$ & $\frac{2\cdot 9\cdot 2!}{8!}$ &
$\frac{2\cdot 11\cdot 3!}{9!}$
\end{tabular} 
& 
\begin{tabular}{cccc}
2 & 3 & 4 & 5 \\ \hline 
$\frac{4\cdot 2\cdot 1!}{3\cdot 4!}$ & $\frac{4\cdot 3\cdot 2!}{5\cdot 5!}$ &
$\frac{4\cdot 4\cdot 3!}{7\cdot 6!}$ & $\frac{4\cdot 5\cdot 4!}{9\cdot 7!}$ \\ 
$\frac{4\cdot 2\cdot 7\cdot 0!}{3\cdot 5!}$ & 
$\frac{4\cdot 3\cdot 10\cdot 1!}{5\cdot 6!}$ & 
$\frac{4\cdot 4\cdot 13\cdot 2!}{7\cdot 7!}$ &
$\frac{4\cdot 5\cdot 16\cdot 3!}{9\cdot 8!}$ \\ 
$\frac{6\cdot 4\cdot 4\cdot 0!}{3\cdot 6!}$ & 
$\frac{6\cdot 4\cdot 9\cdot 1!}{5\cdot 7!}$ & 
$\frac{6\cdot 4\cdot 16\cdot 2!}{7\cdot 8!}$ & 
$\frac{6\cdot 4\cdot 25\cdot 3!}{9\cdot 9!}$ \\
--- & $\frac{4\cdot 3\cdot 2!}{5\cdot 6!}$ & 
$\frac{4\cdot 4\cdot 3!}{7\cdot 7!}$ & $\frac{4\cdot 5\cdot 4!}{9\cdot 8!}$ \\
--- & $\frac{8\cdot 3\cdot 2!}{5\cdot 7!}$ & 
$\frac{8\cdot 4\cdot 3!}{7\cdot 8!}$ & $\frac{8\cdot 5\cdot 4!}{9\cdot 9!}$
\end{tabular}
\end{tabular}
\caption{Summary of results for other moments for $\beta=1,2,4$, some of which
were obtained numerically.}
\label{tab:table}
\end{table}

In Table~\ref{tab:table} we summarize some moments expressed in
Eqs.~(\ref{eq:x12x2N=2}), (\ref{eq:x12x2N=3}), (\ref{eq:q211}), (\ref{eq:x21}),
and others that were obtained numerically. They can easily be generalized to 
\begin{equation}
\left\langle q_1^2\,q_2 \right\rangle^{(\beta)} = 
\left( \frac{\beta}{2} \right)^3 
\frac{ \left( \frac{\beta N}{2}-2\right)! }{\left( \frac{\beta N}{2}\right)! }
\frac{ \left[ \frac{\beta (N+1)}{2}-1\right]! }
{\left[ \frac{\beta(N+1)}{2}\right]! }
K_N^{(\beta)}(2,1,0,0,\ldots,0),
\end{equation}
with 
\begin{equation}
K_N^{(\beta)}(2,1,0,0,\ldots,0) = \frac{2N}{N+2}; 
\end{equation}
\begin{equation}
\label{eq:q3q1}
\left\langle q_1^3\,q_2\right\rangle^{(\beta)} = 
\left( \frac{\beta}{2} \right)^4 
\frac{\left(\frac{\beta N}{2}-3\right)!}{\left(\frac{\beta N}{2}\right)!} \,
\frac{\left[\frac{\beta(N+1)}{2}-1\right]!}{\left[\frac{\beta(N+1)}{2}\right]!}
\, K_N^{(\beta)}(3,1,0,0,\ldots,0) , 
\end{equation}
where
\begin{equation}
K_N^{(\beta)}(3,1,0,0,\ldots,0) = 
\frac{2\left[3N^2+N+(2-\beta/2) \right]}
{(N+3)(N+2)};
\end{equation}
\begin{equation}
\left\langle q_1^3\,q_2^2\right\rangle^{(\beta)} = 
\left( \frac{\beta}{2} \right)^5 
\frac{\left(\frac{\beta N}{2}-3\right)!}{\left(\frac{\beta N}{2}\right)!} \,
\frac{\left[\frac{\beta(N+1)}{2}-2\right]!}{\left[\frac{\beta(N+1)}{2}\right]!}
\, K_N^{(\beta)}(3,2,0,0,\ldots,0) , 
\end{equation}
with
\begin{equation}
K_N^{(\beta)}(3,2,0,0,\ldots,0) = 
\frac{6N^2(2N+1)}{(N+4)(N+3)(N+2)};
\end{equation}
\begin{equation}
\label{eq:q2q1q1}
\left\langle q_1^2 q_2 q_3 \right\rangle^{(\beta)} = 
\left(\frac{\beta}{2}\right)^4 
\frac{ \left( \frac{\beta N}{2}-2\right)! }{\left( \frac{\beta N}{2}\right)! }
\frac{ \left[ \frac{\beta (N+1)}{2}-1\right]! }
{\left[ \frac{\beta(N+1)}{2}\right]! }
\frac{ \left[ \frac{\beta (N+2)}{2}-1\right]! }
{\left[ \frac{\beta(N+2)}{2}\right]! }
K_N^{(\beta)}(2,1,1,0,\ldots,0),
\end{equation}
where 
\begin{equation}
K_N^{(\beta)}(2,1,1,0,\ldots,0) = 
\frac{2N}{N+3},
\end{equation}
and 
\begin{equation}
\left\langle q_1^2 q_2^2 q_3 \right\rangle^{(\beta)} = 
\left(\frac{\beta}{2}\right)^5 
\frac{ \left( \frac{\beta N}{2}-2\right)!}{\left( \frac{\beta N}{2}\right)!}
\frac{ \left[ \frac{\beta (N+1)}{2}-2\right]! }
{\left[ \frac{\beta(N+1)}{2}\right]! }
\frac{ \left[ \frac{\beta (N+2)}{2}-1\right]! }
{\left[ \frac{\beta(N+2)}{2}\right]! }
K_N^{(\beta)}(2,2,1,0,\ldots,0)
\end{equation}
with 
\begin{equation}
K_N^{(\beta)}(2,2,1,0,\ldots,0) = 
\frac{2N(2N+1)}{(N+4)(N+3)}.
\end{equation}

The results of Eqs.~(\ref{eq:q2q2}), (\ref{eq:q3q1}) and (\ref{eq:q2q1q1}) were
reported in Ref.~\onlinecite{Mucciolo} for $\beta=2$.


\section{Conclusions}
\label{sec:conclusions}

We have calculated generalized joint moments of proper delay times for an
arbitrary number of channels $N$ and any symmetry $\beta=1,2$ and $4$, which are
needed to quantify transport properties, or their fluctuations, through
ballistic open systems. This was done by reducing the calculation of the
negative moments of the generalized Laguerre distribution to simpler formulas,
which are easier to manage analytically. As an important result we show that the
$k$th moment of a proper delay time differs from that of the partial delay time,
where the difference comes from the level repulsion of the proper delay times.
From our results, also, we were able to generalize the distribution of the
partial times to any symmetry. Our general expressions reproduce the existing
results for particular cases and those obtained for individual proper and
partial delay times. Also, we obtained the mean and variance of the Wigner time
delay for arbitraries $N$ and $\beta$, which reproduces the known results for
$\beta=2$. Although we regarded perfect coupling of the system to the
open channels, we hope that our results encourage further calculations to
include an imperfect coupling, as was done for partial delay times. 

\acknowledgments

M. Mart\'inez-Mares and J. C. Garc\'ia are grateful with the Sistema Nacional de
Investigadores, Mexico. M. Mart\'inez-Mares is also grateful with M. A.
Torres-Segura for her encouragement. A. M. Mart\'inez-Arg\"uello thanks
CONACyT, Mexico for financial support.

\appendix

\section{Properties associated to $F_{\theta_n}$}
\label{appendix}

There are some relations between the Gamma, Poisson and Negative Binomial 
probability distributions that are summarized in the following propositions:

\begin{proposition}
\label{relaciones}
For $b,b_1,b_2>0$, $k,m,n\in\mathbb{Z_+}$, and $x\ge 0$,
\begin{eqnarray}
& & F_{n,b}(x) =
\int_0^x f_{n,b} (t)\,\mathrm{d}t
= \sum_{k=n}^{\infty} P_{\mathrm{Poisson}}(bx,k), 
\label{prop:1a} \\
& & \int_{0}^{x} f_{m,b_2}(t) P_{\mathrm{Poisson}}(b_1t,k) \mathrm{d}t = 
NB_{m,p}(k) F_{m+k,b_1+b_2}(x), 
\label{prop:1b}
\end{eqnarray}
where $p=b_2/(b_1+b_2)$ and $P_{\mathrm{Poisson}}(\lambda,k)$ is the Poisson
distribution with parameter $\lambda$ given by
\begin{equation}
P_{\mathrm{Poisson}}(\lambda,k) = \frac{\lambda^k}{k!} 
\mathrm{e}^{-\lambda}, \quad\mbox{with}\quad k=0, 1,\ldots\, ,
\end{equation}
and $NB_{m,p}(k)$ is the Negative Binomial distribution~\cite{Abramowitz} with
parameters $m\in\mathbb{N}$ and $p\in[0,1]$:
\begin{equation}
\label{eq:NBinomial}
NB_{m,p}(k) = \left( \begin{array}{c} m-1+k \\ m-1 \end{array} \right)
\left( 1-p \right)^k p^m, \quad\mbox{with}\quad k=0,1,\ldots\,.
\end{equation}
\end{proposition}

\begin{proof}
The proof of (\ref{prop:1a}) is based on the following identity, which is proved
by integrating by parts and induction on $n$ (see Chap. 4, exercise 26, p. 200
of Ref.~\onlinecite{Ross}), 
\begin{equation}
\frac{1}{n!} \int_{x}^{\infty}t^{n} \mathrm{e}^{-t} \mathrm{d}t = 
\sum_{k=0}^{n} \mathrm{e}^{-x} \frac{x^k}{k!} .
\end{equation}
Therefore, with the change of variables $u=bt$ we have that 
\begin{equation}
F_{n,b}(x) = 1 -\frac{1}{(n-1)!} 
\int_{bx}^{\infty} u^{n-1} \mathrm{e}^{-u} \mathrm{d}u =
\sum_{k=n}^{\infty}\frac{(bx)^{k}}{k!}
\mathrm{e}^{-bx} = \sum_{k=n}^{\infty} P_{\mathrm{Poisson}}(bx,k).
\end{equation}

To prove (\ref{prop:1b}), let $p=b_2/(b_1+b_2)$. Then, 
\begin{equation}
\int_{0}^{x} f_{m,b_2}(t) P_{\mathrm{Poisson}}(b_1t,k) \mathrm{d}t = 
\frac{b_2^m b_1^k}{(m-1)!k!} 
\int_{0}^{x} t^{m+k-1} \mathrm{e}^{-(b_1+b_2)t} \mathrm{d}t = 
NB_{m,p}(k) F_{m+k,b_1+b_2}(x).
\end{equation}
\end{proof}

Let
\begin{equation}
F^{(n)}_{\theta_{n}}(x):= 
\int_{\Delta_{n}(x)} \, \prod_{i=1}^{n}f_{a_i,b_i}(t_i)\,
\mathrm{d}t_i, 
\end{equation}
where $\Delta_n(x)=\left\{\left(t_1,\ldots,t_n\right)\,:\,0\le t_1\le
t_2\le\cdots\le t_n\le x\right\}$. By the Cavallieri Principle, the functions
$F^{(n)}_{\theta_n}(x)$ satisfy the recurrence relation
\begin{eqnarray}
\label{cavallieri}
F^{(n)}_{\theta_n}(y)=\int_{0}^{y}f_{a_n,b_n}(x)
F^{(n-1)}_{\theta_{n-1}}(x) \mathrm{d}x.
\end{eqnarray}
Also, for this function we have the following properties:

\begin{proposition}
\label{relaciones2} 
If $n\ge 2$, all the $a_i\in\mathbb{N}$, $g_k:=\sum_{i=1}^{k}b_i$, and
$p_k:=b_k/g_k$, then
\begin{eqnarray}
F^{(1)}_{\theta_1}(y) & = & \sum_{k=a_1}^{\infty} P_{{\rm Poisson}}(b_1y,k),
\quad F_{\theta_1} = 1, \label{propiedadbasica1} \\
F^{(n)}_{\theta_n}(x) & = & \sum_{\mathbf{A}^{(n-1)}}
\prod_{j=2}^{n} N B_{a_j,p_j} \big(L_{j-1} \big) F_{ a_n+L_{n-1},g_n}(x),
\label{propiedadbasica2} \\
F_{\theta_n} & = & \sum_{\mathbf{A}^{(n-1)}}
\prod_{j=2}^{n} N B_{a_j,p_j} \left(L_{j-1} \right), 
\label{propiedadbasica3}
\end{eqnarray}
where $L:=\sum_{i=1}^n\ell_i$, for $\ell_i\in\mathbb{N}$, 
$\sum_{\mathbf{A}^{(n-1)}}$ means summation over all 
$\ell_i\in\mathbf{A}^{(n-1)}$, with 
\begin{equation}
\label{cartesianos}
\mathbf{A}^{(n)}:=\{ \ell_1,\ldots,\ell_n:\ell_i\ge a_i,\ell_i\in\mathbb{N},
i=1,\ldots,n\}.
\end{equation}
\end{proposition}

\begin{proof}
The first formula in (\ref{propiedadbasica1}), is exactly (\ref{prop:1a}) of
Proposition~\ref{relaciones}. The second is clear since $F^{(1)}_{\theta_1}$ is
a probability distribution function. To prove (\ref{propiedadbasica2}), we
proceed by induction on $n$. For $n=2$, (\ref{propiedadbasica2}) is a direct
consequence of Proposition~\ref{relaciones}, since $a_1\in\mathbb{N}$, and the
Monotone Convergence Theorem to interchange the series and the integral. As
inductive hypothesis, let us assume that
(\ref{propiedadbasica2}) holds for $n-1$, with some $n\ge 3$, and try to prove
it for $n$. With this hypothesis we write Eq.~(\ref{cavallieri}) as
\begin{eqnarray}
F^{(n)}_{\theta_{n}}(y) & = &
\sum_{\mathbf{A}^{(n-2)}} \prod_{j=2}^{n-1} NB_{a_j,p_j} 
\left(L_{j-1}\right) \int_{0 }^{y}f_{a_n,b_n}(x)
F_{a_{n-1}+L_{n-2},g_{n-1}}(x) \mathrm{d} x 
\nonumber \\ & = &
\sum_{\mathbf{A}^{(n-2)}} \prod_{j=2}^{n-1} NB_{a_j,p_j}\left(L_{j-1} \right)
\sum_{k=a_{n-1}+L_{n-2}}^{\infty} NB_{a_{n},p_{n}}\big(k\big)
F_{a_n+k,g_n}
\big(y\big) ,
\end{eqnarray}
where at the last equality we first used (\ref{prop:1a}) and then
(\ref{prop:1b}) with $m=a_n$, $k=L_{n-2}$, $p=p_n$, taking into
account that $p_n=b_n/(b_n+g_{n-1})$. If we make the change of variables
$\ell_{n-1}=k-L_{n-2}$ we write the last equation as
\begin{eqnarray}
F^{(n)}_{\theta_{n}}(y) & = &
\sum_{\mathbf{A}^{(n-2)}} \prod_{j=2}^{n-1} NB_{a_j,p_j}
\left(L_{j-1}\right) 
\sum_{\ell_{n-1}=a_{n-1}}^{\infty} NB_{a_n,p_n} 
\big(\ell_{n-1}+L_{n-2}\big)
F_{a_n+\ell_{n-1}+L_{n-2},g_{n}}\big(y\big) 
\nonumber \\ & = & 
\sum_{\mathbf{A}^{(n-1)}} \prod_{j=2}^{n}
NB_{_j,p_j} \left(L_{j-1}\right) F_{a_n+L_{n-1},g_n}(y) .
\end{eqnarray}
At the last step we used that the sets $\mathbf{A}^{(n)}$ are cartesian products
[see~(\ref{cartesianos})] in fact, 
\begin{equation}
\mathbf{A}^{(n-1)} = \mathbf{A}^{(n-2)}\times\{a_{n-1},1+a_{n-1},
2+a_{n-1},\dots\}
\end{equation}
This finishes the proof of \ref{propiedadbasica2}.

The proof of (\ref{propiedadbasica3}) is based on the Monotone Convergence
Theorem. Indeed, 
\begin{equation}
\label{eq:Ftheta}
F_{\theta_n} = \lim_{x\to\infty} F^{(n)}_{\theta_n}(x) =
\sum_{\mathbf{A}^{(n-1)}} \prod_{j=2}^{n}
NB_{a_j,p_j} \big(L_{j-1}\big) \lim_{
x\to\infty} F_{a_n+L_{n-1},g_n}(x) = 
\sum_{\mathbf{A}^{(n-1)}} \prod_{j=2}^{n} NB_{a_j,p_j}
\big(L_{j-1}\big) \nonumber
\end{equation}
\end{proof}

\begin{remark}
It is not necessary, for Proposition~\ref{relaciones2} to hold, that the
$b_i$'s to belong to $\mathbb{N}$, just that they are positive.
\end{remark}




\begin{thebibliography}{100}
\bibitem{footnote}
The negative binomial distribution can be used only when $N$ is an even number.
\end{thebibliography}
\end{document}